\documentclass[11pt]{amsart}

\usepackage{amsmath,amsthm,amsfonts,amscd}
\usepackage{amssymb}
\usepackage{bm}
\usepackage{mathrsfs}
\usepackage[pdftex]{graphicx}
\usepackage[mathcal]{eucal}
\usepackage{verbatim}
\usepackage{url}
\usepackage{color}
\usepackage{supertabular}
\usepackage{multirow}
\usepackage{algorithmic}

\newtheorem{thm}{Theorem}
\newtheorem{claim}{Claim}

\newtheorem{lemma}[thm]{Lemma}

\newtheorem{definition}[thm]{Definition}

\numberwithin{equation}{section}

\newenvironment{reftheorem}[1]{\medskip\parindent 0pt{\bf Theorem \ref{#1}}\em }{\vspace{1em}}

\DeclareMathAlphabet{\mathsfsl}{OT1}{cmss}{m}{sl}

	{\makebox{\phantom{}} \\ %
		\textsc{Input:}\begin{itemize}}%
	{\end{itemize}}

	{\textsc{Output:}\begin{itemize}}%
	{\end{itemize}}

	{\textsc{Procedure:}\begin{enumerate}}%
	{\end{enumerate}}

\renewcommand{\phi}{\varphi}

\newcommand{\Expect}{\operatorname{\mathbb{E}}}

\newcommand{\vct}[1]{\bm{#1}}
\newcommand{\mtx}[1]{\bm{#1}}

\newcommand{\norm}[1]{\left\Vert {#1} \right\Vert}

\newcommand{\pnorm}[2]{\norm{#2}_{#1}}

\newcommand{\nerr}[1]{\left\| #1 \right\|}

\newcommand\rds{{\oplus_{{\rm r}}}}
\newcommand\sdp{{\ltimes_{{\rm r}}}}
\newcommand\hp{{\odot}}

\newcommand{\signal}{{\mathbf x}}

\title{Approximate Sparse Recovery: Optimizing Time and Measurements}
\author[Gilbert, Li, Porat, Strauss]{A.\ C.\ Gilbert, Y.\ Li, E.\ Porat, and M.\ J.\ Strauss}

\thanks{Gilbert is with the Department of Mathematics, The University of 
  Michigan at Ann Arbor.  E-mail: \url{annacg}\url{@umich.edu}.  Li is with the
  Department of Electrical Engineering and Computer Science, The University of 
  Michigan at Ann Arbor.  E-mail: \url{leeyi@umich.edu}.  Porat is with the 
  Department of Computer Science, Bar-Ilan University.  E-mail: 
  \url{porately@cs.biu.ac.il}.  Strauss is
  with the Department of Mathematics and the Department of Electrical
  Engineering and Computer Science, The University of Michigan at Ann
  Arbor.  E-mail: \url{martinjs@umich.edu}.
}

\begin{document}
\maketitle

\begin{abstract}
An {\em approximate sparse recovery} system consists of
parameters $k,N$, an $m$-by-$N$ {\em measurement matrix},
$\mtx{\Phi}$, and a decoding algorithm, $\mathcal{D}$.  Given a vector, $\signal$,
the system approximates $\signal$ by 
$\widehat \signal=\mathcal{D}(\mtx{\Phi} \signal)$,
which must satisfy
$\nerr{\widehat \signal -   \signal}_2\le C \nerr{\signal - \signal_k}_2$, where
$\signal_k$ denotes the optimal $k$-term approximation to $\signal$.  For each
vector $\signal$, the system must succeed with probability at least 3/4.  Among the goals in designing such systems are minimizing the number $m$ of measurements and the runtime of the
decoding algorithm, $\mathcal{D}$.

In this paper, we give a system with $m=O(k \log(N/k))$
measurements---matching a lower bound, up to a constant factor---and
decoding time $O(k\log^c N)$, matching a lower bound up
to $\log(N)$ factors.
We also consider the encode time ({\em i.e.}, the time to multiply
$\mtx{\Phi}$ by $x$), the time to update measurements ({\em i.e.}, the time to
multiply $\mtx{\Phi}$ by a 1-sparse $x$), and the robustness and stability of the algorithm (adding noise before and after the measurements).  Our
encode and update times are optimal up to $\log(N)$ factors.
The columns of $\mtx{\Phi}$ have at most $O(\log^2(k)\log(N/k))$ non-zeros, each of which can be found in constant time.  If $\signal$ is an exact $k$-sparse signal and $\nu_1$ and $\nu_2$ are arbitrary vectors (regarded as noise), then, setting $\widehat \signal = \mathcal{D}(\Phi (\signal+\nu_1)+\nu_2)$, we get
\[
\nerr{\widehat \signal - \signal}_2
    \le  2 \nerr{\nu_1}_2 + 
            \log(k)\frac{\nerr{\nu_2}_2}{\nerr{\Phi}_{2\leadsto 2}},
\]
where ${\nerr{\Phi}_{2\leadsto 2}}$ is a natural scaling factor that
makes our result comparable with previous results.  (The $\log(k)$
factor above, improvable to $\log^{1/2+o(1)} k$, makes our result
(slightly) suboptimal when $\nu_2\ne 0$.)  We also extend our recovery
system to an FPRAS.
\end{abstract}

\allowdisplaybreaks

\section{Introduction}

Tracking heavy hitters in high-volume, high-speed data streams~\cite{CCF}, monitoring changes in data streams~\cite{CM03:Whats-Hot}, designing pooling schemes for biological tests~\cite{Erlich09:DNA-Sudoku} (e.g., high throughput sequencing, testing for genetic markers), localizing sources in sensor networks~\cite{Brady05:Fiber-graycode,Brady06:Optics-grouptesting} are all quite different technological challenges, yet they can all be expressed in the same mathematical formulation.  We have a signal $\signal$ of length $N$ that is sparse or highly compressible; i.e., it consists of $k$ significant entries (``heavy hitters'') which we denote by $\signal_k$ while the rest of the entries are essentially negligible.  We wish to acquire a small amount information (commensurate with the sparsity) about this signal in a linear, non-adaptive fashion and then use that information to quickly recover the significant entries.  In a data stream setting, our signal is the distribution of items seen, while in biological group testing, the signal is proportional to the binding affinity of each drug compound (or the expression level of a gene in a particular organism).  We want to recover the identities and values only of the heavy hitters which we denote by $\signal_k$, as the rest of the signal is not of interest.  Mathematically, we have a signal $\signal$ and an $m$-by-$N$ measurement matrix $\mtx{\Phi}$ with which we acquire measurements $\vct{y} = \mtx{\Phi}\signal$, and, from these measurements $\vct{y}$, we wish to recover $\widehat \signal$, with $O(k)$ entries, such that
\[   \pnorm{2}{\signal - \widehat \signal} \leq 
            C \pnorm{2}{\signal - \signal_k}.
\]
Our goal, which we achieve up to constant or log factors in the
various criteria, is to design the measurement matrix $\mtx{\Phi}$ and
the decoding algorithm in an optimal fashion: (i) we take as few
measurements as possible $m = O(k\log (N/k))$, (ii) the decoding
algorithm runs in {\em sublinear} time $O(k\log(N/k))$, and (iii) the
encoding and update times are optimal $O(N \log(N/k))$ and
$O(k\log(N/k))$, respectively.  In order to achieve this, our
algorithm is a randomized algorithm; i.e., we specify a distribution
on the measurement matrix $\mtx{\Phi}$ and we guarantee that, for each
signal, the algorithm recovers a good approximation with high
probability over the choice of matrix.

In the above applications, it is important both to take as few measurements as possible and to recover the heavy hitters extremely efficiently.  Measurements correspond to physical resources (e.g., memory in data stream monitoring devices, number of screens in biological applications) and reducing the number of necessary measurements is critical these problems.  In addition, these applications require efficient recovery of the heavy hitters---we test many biological compounds at once, we want to quickly identify the positions of entities in a sensor network, and we cannot afford to spend computation time proportional to the size of the distribution in a data stream application.  Furthermore, Do Ba, et al.~\cite{BaIndykPriceWoodruff:2010} give a lower bound on the number of measurements for sparse recovery $\Omega(k\log(N/k))$. There are polynomial time algorithms~\cite{NT08,CRT06:Stable-Signal,RI08} meet this lower bound, both with high probability for each signal and the stronger setting, with high probability for all signals\footnote{albeit with different error guarantees and different column sparsity depending on the error metric.}.  Previous sublinear time algorithms, whether in the ``for each'' model~\cite{CCF,CM06:Combinatorial-Algorithms} or in the ``for all'' model~\cite{GSTV07:HHS}, however, used several additional factors of $\log(N)$ measurements.  We summarize the previous sublinear algorithms in the ``for each'' signal model in Figure~\ref{fig:previouswork}.  The column sparsity denotes how many 1s there are per column of the measurement matrix and determines both the decoding and measurement update time and, for readability, we suppress $O(\cdot)$.  The approximation error signifies the metric we use to evaluate the output; either the $\ell_2$ or $\ell_1$ metric.  In this paper, we focus on the $\ell_2$ metric.
\begin{figure}
{\footnotesize
\begin{tabular}{|c|c|c|c|c|c|c|c|}
\hline
Paper		& No. Measurements &	Encode time &	Column sparsity/ &	Decode time & 	Approx. error\\
 & & & Update time & & \\
\hline 
\cite{Don06:Compressed-Sensing,CRT06:Stable-Signal}
                & $k\log( N/k)$    &  $Nk\log(N/k)$ &   $k\log(N/k)$     &      $\ge N$     &  $\ell_2 \le (1/\sqrt{k}) \ell_1$\\
\cite{CCF,CM06:Combinatorial-Algorithms}  &	$k \log^c N$ &	$N \log^c N$ &	$\log^c N$  &	$k \log^c N$ & 	$\ell_2 \le C \ell_2$\\
\cite{CM03b} &	$k \log^c N$ &	$N \log^c N$ &	$\log^c N$ &	$k \log^c N$ & 	$\ell_1 \le C \ell_1$\\
\hline 
This paper	 &	$k \log(N/k)$ & 	$N \log^c N$ & 	$\log^c N$   & $k\log^c N$ & $\ell_2 \le C \ell_2$\\
\hline
\end{tabular}
\caption{\footnotesize Summary of the best previous results and the result obtained in this paper. }
\label{fig:previouswork}
}
\end{figure}

We give a joint distribution over measurement matrices and sublinear time recovery algorithms that meet this lower bound (up to constant factors) in terms of the number of measurements and are within $\log(k)$ factors of optimal in the running time and the sparsity of the measurement matrix.  
\begin{thm}
	\label{thm:plain}
	There is a joint distribution on matrices and algorithms, with
    suitable instantiations of anonymous constant factors, such that, 
    given measurements $\mtx{\Phi}\signal = \vct{y}$, the algorithm returns
      $\widehat x$ and approximation error
\[ \pnorm{2}{\signal - \widehat \signal} \leq
    2 \pnorm{2}{\nu_1}
\]
with probability $3/4$.  The algorithm runs in time $O(k \log^c(N))$ and $\mtx{\Phi}$ has $O(k \log(N/k))$ rows.
\end{thm}
Furthermore, our algorithm is a fully polynomial randomized approximation scheme.
\begin{thm}
	\label{thm:FPRAS}
    There is a joint distribution on matrices and algorithms, with
    suitable instantiations of anonymous constant factors (that
    may depend on $\epsilon$, such that, given measurements
      $\mtx{\Phi}\signal = \vct{y}$, the algorithm returns
      $\widehat x$ and approximation error
\[ \pnorm{2}{\signal - \widehat \signal} \leq
    (1+\epsilon)\pnorm{2}{\nu_1}
\]
with probability $3/4$.  The algorithm runs in time $O((k/\epsilon)
    \log^c(N))$ and $\mtx{\Phi}$ has $O((k/\epsilon) \log(N/k))$ rows.
\end{thm}
Finally, our result is robust to corruption of the measurements by an
arbitrary noise vector $\nu_2$, which is an important feature for such
applications as high throughput screening and other physical
measurement systems.  (It is less critical for digital measurement
systems that monitor data streams in which measurement corruption is
less likely.)  When $\nu_2\ne 0$, our error dependence is on $\nu_2$
is suboptimal by the factor $\log(k)$ (improvable to $\log^{1/2+o(1)}
k)$.  Equivalently, we can use $\log(k)$ times more measurements to
restore optimality.
\begin{thm}
	\label{thm:robustFPRAS}
	There is a joint distribution on matrices and algorithms, with
    suitable instantiations of anonymous constant factors (that
    may depend on $\epsilon$), such that, given measurements
      $\mtx{\Phi}\signal + \nu_2 = \vct{y} + \nu_2$, the algorithm returns
      $\widehat x$ and approximation error
\[ \pnorm{2}{\signal - \widehat \signal} \leq
    (1+\epsilon)\pnorm{2}{\nu_1} + 
            \epsilon\log(k)\frac{\nerr{\nu_2}_2}{\nerr{\Phi}_{2\leadsto 2}}
\]
with probability $3/4$.  The algorithm runs in time $O((k/\epsilon)
    \log^c(N))$ and $\mtx{\Phi}$ has $O(k/\epsilon \log(N/k))$ rows.
\end{thm}

Previous sublinear algorithms begin with the observation that if a signal consists of a single heavy hitter, then the trivial encoding of the positions 1 through $N$ with $\log(N)$ bits, referred to as a bit tester, can identify the position of the heavy hitter.  The second observation is that a number of hash functions drawn at random from a hash family are sufficient to isolate enough of the heavy hitters, which can then be identified by the bit tester.  Depending on the type of error metric desired, the hashing matrix is pre-multiplied by random $\pm 1$ vectors (for the $\ell_2$ metric) in order to estimate the signal values.  In this case, the measurements are referred to as the {\sc Count Sketch} in the data stream literature~\cite{CCF} and, without the premultiplication, the measurements are referred to as {\sc Count Median}~\cite{CM03b,CM06:Combinatorial-Algorithms} and give $\ell_1\leq C \ell_1$ error guarantees.  In addition, the sublinear algorithms are typically greedy, iterative algorithms that recover portions of the heavy hitters with each iteration or that recover portions of the $\ell_2$ (or $\ell_1$) energy of the residual signal.

We build upon the {\sc Count Sketch} design but incorporate the following algorithmic innovations to ensure an optimal number of measurements:
\begin{itemize}
	\item With a random assignment of $N$ signal positions to $O(k)$ measurements, we need to encode only $O(N/k)$ positions, rather than $N$ as in the previous approaches.  So, we can reduce the domain size which we encode.
	\item We use a good error-correcting code (rather than the trivial identity code of the bit tester).
	\item Our algorithm is an iterative algorithm but maintains a {\em compound} invariant: the number of un-discovered heavy hitters decreases at each iteration while, simultaneously, the required error tolerance and failure probability become more stringent.  Because there are fewer heavy hitters to find at each stage, we can use more measurements to meet more stringent guarantees.
\end{itemize}

In Section~\ref{sec:preliminaries} we detail the matrix algebra we use to describe the measurement matrix distribution which we cover in Section~\ref{sec:recoverysystem}, along with the decoding algorithm.  In Section~\ref{sec:analysis}, we analyze the foregoing recovery system.

\section{Preliminaries}
\label{sec:preliminaries}

\subsection{Vectors}
Let $\signal$ denote a vector of length $N$.  For each $k \le N$, let
$\signal_k$ denote either the usual $k$'th component of $\signal$ or
the signal of length $N$ consisting of the $j$
largest-magnitude terms in $\signal$; it will be clear from context.
The signal $\signal_k$ is the best
$k$-term representation of $\signal$.  The energy of a signal
$\signal$ is $\pnorm{2}{\signal}^2 = \sum_{i=1}^N |\signal_i|^2$.


\subsection{Matrices}
\begin{figure}[h!]
	\centering
	\begin{tabular}{|c|c|l|l|} \hline
		operator & name & input & output dimensions and construction \\ \hline
		\multirow{2}{*}{$\oplus_r$} & {row direct sum} & 
		      $\mtx{A} \colon r_1 \times N$ & $\mtx{M} \colon (r_1 + r_2) \times N$ \\ 
		&   & $\mtx{B} \colon r_2 \times N$ & $\mtx{M}_{i,j} = 
		            \begin{cases} \mtx{A}_{i,j}, & 1 \leq i \leq r_1 \\
		                          \mtx{B}_{i-r_1,j}, & 1+r_1 \leq i \leq r_2 
		  			 \end{cases}$ \\ \hline
		\multirow{2}{*}{$\odot$} & {element-wise product} & 
		       $\mtx{A} \colon r \times N$ & $\mtx{M} \colon r \times N$ \\ 
		&    & $\mtx{B} \colon r \times N$ & $\mtx{M}_{i,j} = \mtx{A}_{i,j}\mtx{B}_{i,j}$ \\ \hline
		\multirow{2}{*}{$\sdp$} & {semi-direct product} & 
		       $\mtx{A} \colon r_1 \times N$ & $\mtx{M} \colon (r_1 r_2) \times N$ \\ 
		&    & $\mtx{B} \colon r_2 \times h$ & $\mtx{M}_{i+(k-1)r_2,\ell} = 
		          \begin{cases} 0, & \mtx{A}_{k,\ell} = 0 \\
                  \mtx{A}_{k,\ell}\mtx{B}_{i,j}, & \mtx{A}_{k,\ell} = \text{ $j$th nonzero in row $\ell$}
			  	  \end{cases}$ \\ \hline
	\end{tabular}
	\caption{Matrix algebra used in constructing an overall measurement matrix.  The last column contains both the output dimensions of the matrix operation and its construction formula.}
	\label{table:matrixalgebra}
\end{figure}

In order to construct the overall measurement matrix, we form a number of different types of combinations of constituent matrices and to facilitate our description, we summarize our matrix operations in Table~\ref{table:matrixalgebra}. The matrices that result from all of our matrix operations have $N$ columns and, with the exception of the semi-direct product of two matrices $\sdp$, all operations are performed on matrices $\mtx{A}$ and $\mtx{B}$ with $N$ columns.  A full description can be found in the Appendix.


\section{Sparse recovery system}
\label{sec:recoverysystem}
In this section, we specify the measurement matrix and detail the decoding algorithm. 
\subsection{Measurement matrix}
\label{sec:matrix}

The overall measurement matrix, $\mtx{\Phi}$, is a multi-layered matrix with entries in $\{-1, 0, +1\}$. At the highest level, $\mtx{\Phi}$ consists of a random permutation matrix $\mtx{P}$ left-multiplying  the row direct sum of
$(\lg(k))$ summands, $\mtx{\Phi}^{(j)}$, each of which is used in a separate iteration of the decoding algorithm.  Each summand $\mtx{\Phi}^{(j)}$ is the row direct sum of two separate matrices, an {\em identification} matrix, $\mtx{D}^{(j)}$, and an {\em estimation} matrix, $\mtx{E}^{(j)}$.
\[  
    \mtx{\Phi} = \mtx{P} \begin{bmatrix}
		\mtx{\Phi}^{(1)} \\ \hline
	    \mtx{\Phi}^{(2)} \\ \hline
	    \vdots \\ \hline
		\mtx{\Phi}^{(\lg(k))} \\
   		\end{bmatrix}
	\qquad\text{where $\mtx{\Phi}^{(j)} = \mtx{E}^{(j)} \rds \mtx{D}^{(j)}$}.
\]

In iteration $j$, the identification matrix $\mtx{D}^{(j)}$ consists of the row
direct sum of $O(j)$ matrices, all chosen independently from the same
distribution.  We construct the distribution $(\mtx{C}^{(j)} \sdp \mtx{B}^{(j)}) \hp \mtx{S}^{(j)}$ as follows:
\begin{itemize}
\item For $j=1,2,\ldots,\lg(k)$, the matrix $\mtx{B}^{(j)}$ is a
  Bernoulli matrix with dimensions
  $kc^j$-by-$N$, where $c$ is an appropriate constant $1/2< c <1$.  Each entry
  is $1$ with probability $\Theta\left( 1/(k c^j) \right)$ and
  zero otherwise. Each row is a pairwise independent family and the
  set of row seeds is fully independent.
\item The matrix $\mtx{C}^{(j)}$ is an encoding of positions by an error-correcting code
  with constant rate and relative distance.  That is,
  fix an error-correcting code and encoding/decoding algorithm that
  encodes messages of $\Theta(\log\log N)$ bits into longer codewords,
  also of length $\Theta(\log\log N)$, and can correct a constant
  fraction of errors.  
  The $i$'th column of $\mtx{C}^{(j)}$ is the direct
  sum of $\Theta(\log\log N)$ copies of $1$ with the direct sum of $E(i_1),
  E(i_2), \dots$, where $i_1,i_2,\dots$ are blocks of $O(\log\log N)$
  bits each whose concatenation is the binary expansion of $i$ and
  $E(\cdot)$ is the encoding function for the error-correcting code. The
  number of columns in $\mtx{C}^{(j)}$ matches the maximum number of
  non-zeros in $\mtx{B}^{(j)}$, which is approximately the expected
  number, $\Theta\left(c^j N/k\right)$, where $c<1$. The number of rows
  in $\mtx{C}^{(j)}$ is 
    the logarithm of the number of columns, since the process of breaking the binary expansion of index $i$ into blocks has rate 1 and encoding by $E(\cdot)$ has constant rate.
\item The matrix $\mtx{S}^{(j)}$ is a pseudorandom sign-flip matrix. Each row is a pairwise independent family of uniform $\pm 1$-valued random variables. The sequence of seeds for the rows is a fully independent family.
The size of $\mtx{S}^{(j)}$ matches the size of $\mtx{C}^{(j)} \sdp \mtx{B}^{(j)}$.
\end{itemize}
Note that error correcting encoding often is accomplished by a
matrix-vector product, but we are {\em not} encoding a linear
error-correcting code by the usual generator matrix process.  Rather,
our matrix explicitly lists all the codewords.  The code may be
non-linear.

The identification matrix at iteration $j$ is of the form
\[   \mtx{D}^{(j)} = \begin{bmatrix}
					\Big[(\mtx{C}^{(j)} \sdp \mtx{B}^{(j)})
					    \hp \mtx{S}^{(j)}\Big]_1 \\ \hline
					\hdots \\ \hline
					\Big[(\mtx{C}^{(j)} \sdp \mtx{B}^{(j)})
					    \hp \mtx{S}^{(j)}\Big]_{O(j)} \\
					\end{bmatrix}.
\]

In iteration $j$, the estimation matrix $\mtx{E}^{(j)}$ consists of
the direct sum of $O(j)$ matrices, all chosen independently from the
same
distribution, $\mtx{B'}^{(j)} \hp \mtx{S'}^{(j)}$, so that the estimation matrix at iteration $j$ is of the form
\[   \mtx{E}^{(j)} = \begin{bmatrix}
					\Big[\mtx{B'}^{(j)} \hp \mtx{S'}^{(j)}\Big]_1 \\ \hline
					\hdots \\ \hline
					\Big[\mtx{B'}^{(j)} \hp \mtx{S'}^{(j)}\Big]_{O(j)} \\
					\end{bmatrix}.
\]
The construction of the distribution is similar to that of the
identification matrix, but omits the error-correcting code and uses
different constant factors, etc., for the number of rows compared with
the analogues in the identification matrix.
\begin{itemize}
 \item The matrix $\mtx{B'}^{(j)}$ is Bernoulli with dimensions
 $O(k c^j)$-by-$N$, for appropriate $c$, $1/2<c<1$.  Each entry is $1$
   with probability $\Theta\left(1/(kc^j)\right)$ and zero
   otherwise. Each row is a pairwise independent family and the set of
   seeds is fully independent.
	\item The matrix $\mtx{S'}^{(j)}$ is a pseudorandom sign-flip matrix of the same dimension as $\mtx{B'}^{(j)}$.
Each row of $\mtx{S'}^{(j)}$ is a pairwise independent family of uniform $\pm 1$-valued random variables. The sequence of seeds for the rows is a fully independent family.
\end{itemize}

\subsection{Measurements}

The overall form of the measurements mirrors the structure of the
measurement matrices.  We do not, however, use all of the measurements
in the same fashion.  In iteration $j$ of the algorithm, we use the measurements
$\vct{y}^{(j)} = \mtx{\Phi}^{(j)}\signal$.  As the matrix
$\mtx{\Phi}^{(j)} = \mtx{E}^{(j)} \rds \mtx{D}^{(j)}$, we have a
portion of the measurements $\vct{w}^{(j)} = \mtx{D}^{(j)} \signal$
that we use for identification and a portion $\vct{z}^{(j)} =
\mtx{E}^{(j)} \signal$ that we use for estimation.  The
$\vct{w}^{(j)}$ portion is further decomposed into measurements 
$[\vct{v}^{(j)}, \vct{u}^{(j)}]$
corresponding to the run of $O(\log\log N)$ 1's in $\mtx{C}^{(j)}$ and
measurements corresponding to each of the blocks in the
error-correcting code.   There are $O(j)$ i.i.d. repetitions of
everything at iteration $j$.


\subsection{Decoding}

\label{sec:decoding}
The decoding algorithm is shown in Figure \ref{fig:decodingalg} in the Appendix.

\section{Analysis}
\label{sec:analysis}
In this section we analyze the decoding algorithm for correctness and efficiency.

\subsection{Correctness}
\label{sec:proof}

Let $\signal = \signal_k + \nu_1$ where we assume $\signal$ is
normalized so that $\pnorm{2}{\nu_1} = 1$ and $\signal_k$ is the vector $\signal$ with all but the largest-magnitude $k$ entries zeroed
out.   Our goal is to guarantee an approximation $\widehat \signal$
with approximation error $\pnorm{2}{\signal - \widehat \signal} \leq
(1+\epsilon)\pnorm{2}{\nu_1} + \epsilon\pnorm{2}{\nu_2}$.  But observe
that $\nu_2$ is a different type of object from $\signal$ or $\widehat
\signal$; $\nu_2$ is added to $\mtx{\Phi}\signal$.  For the main theorem to
make sense, therefore, we need to normalize $\mtx{\Phi}$.  We discuss
this now.

Observe that the matrix $\mtx{\Phi}$ can be scaled up by an
arbitrary constant factor $c>1$ which can be undone by the decoding algorithm:  Let $\mathcal{D'}$ be a new decoding algorithm that calls the old decoding algorithm $\mathcal{D}$ as follows:  $\mathcal{D'}(\vct{y})=\mathcal{D}\left(\frac1c
\vct{y}\right)$, so that $\mathcal{D'}(c\mtx{\Phi}\signal+\nu_2)
=\mathcal{D}\left(\mtx{\Phi}\signal+\frac1c\nu_2\right)$.
Thus we can {\em reduce} the effect of $\nu_2$ by an arbitrary factor
$c$ and so citing performance in terms of $\nerr{\nu_2}$ alone is not
sensible.  Note also that $\nu_2$ and $\signal$ are different types of
objects; $\mtx{\Phi}$, as an operator, takes an object of the type of
$\signal$ and produces an object of the type of $\nu_2$.  We
will stipulate that the appropriate norm of $\mtx{\Phi}$ be
bounded by 1, in order to make our results quantitatively comparable
with others.  Our error guarantee is in $\ell_2$ norm, so we should
use a 2-operator norm; i.e.,, $\max\pnorm{2}{\mtx{\Phi}\vct{x}}$
over $\vct{x}$ with $\pnorm{2}{\vct{x}}=1$.  But our algorithm's
guarantee is in the ``for each'' signal model, so we need to modify the norm
slightly.  
\begin{definition}
The $\pnorm{2\leadsto 2}{\mtx{\Phi}}$ norm of a randomly-constructed matrix
$\mtx{\Phi}$ is $\max_{\vct{x}} 
           \Expect\left[\frac{\pnorm{2}{\mtx{\Phi}\vct{x}}}{\vct{x}}\right]$.
the smallest $M$ such that, for all $\vct{x}$ with
$\pnorm{2}{\vct{x}}=1$, we have $\pnorm{2}{\mtx{\Phi}\vct{x}}<M$
except with probability 1/4.
\end{definition}

Now we bound $\pnorm{2 \leadsto 2}{\mtx{\Phi}}$.  Each row
$\rho$ of a Bernoulli$(p)$ matrix with sign flips,
$\mtx{B}\hp\mtx{S}$, satisfies
$\Expect[|\rho\vct{x}|^2]=p\pnorm{2}{x}^2$.  So $1/p$ such rows satisfy
$\pnorm{2}{(\mtx{B}\hp\mtx{S})\vct{x}}^2\le
O\left(\pnorm{2}{\vct{x}}^2\right)$.  Our matrix $\mtx{\Phi}$ repeats
the above $j$ times in the $j$'th iteration, $j\le\log_2(k)$,
and combines it with an error-correcting code matrix of
$\Theta(\log(N/k))$ dense rows.  It follows that 
\[  \pnorm{2\leadsto 2}{\mtx{\Phi}}^2 = O(\log^2(k)\log(N/k)).
\]

We are ready to state the main theorem.

\begin{reftheorem}{thm:robustFPRAS}
        Consider the matrices
        in~Section~\ref{sec:matrix} and the algorithms in
        Section~\ref{sec:decoding} (that share randomness with the
        matrices).
	The joint distribution on those matrices and algorithms, with
        suitable instantiations of anonymous constant factors (that
        may depend on $\epsilon$), are such that, given measurements
          $\mtx{\Phi}\signal + \nu_2 = \vct{y} + \nu_2$, the algorithm returns
          $\widehat \signal$ with approximation error
	\[ \pnorm{2}{\signal - \widehat \signal} \leq
        (1+\epsilon)\pnorm{2}{\nu_1} +
               \epsilon\log(k)\frac{\nerr{\nu_2}_2}{\nerr{\Phi}_{2\leadsto 2}}
	\]
	with probability $3/4$.  The algorithm runs in time $k
        \log^{c} N$ and $\mtx{\Phi}$ has $O(k \log(N/k))$ rows.
\end{reftheorem}

In this extended abstract, we give the proof only for $\epsilon=1$.
Our results generalize in a straightforward way for general
$\epsilon>0$ (roughly, by replacing $k$ with $k/\epsilon$ at the
appropriate places in the proof) and the number of measurements is
essentially optimal in $\epsilon$.  Because our approach builds upon the {\sc Count Sketch} approach in~\cite{CCF}, we omit the proof of intermediary steps that have appeared earlier in the literature.

We maintain the following invariant. At the beginning of iteration
$j$, the residual signal has the form
\begin{equation}  \vct{r}^{(j)} = \signal^{(j)} + \nu_1^{(j)}
  \quad\text{with $\pnorm{0}{\signal^{(j)}} \leq \frac{k}{2^j}$, and
    $\pnorm{2}{\nu_1^{(j)}} \leq 2 - \Big(\frac{3}{4}\Big)^j$}
\tag{{\sc Loop Invariant}}
\end{equation}
except with probability $\frac14(1 - (\frac12)^j)$, where
$\pnorm{0}{\cdot}$ is the number of non-zero entries.  The vector
$\signal^{(j)}$ consists of residual elements of $\signal_k$.
Clearly, maintaining the invariant is sufficient to prove the overall
result.  In order to show that the algorithm maintains the loop
invariant, we demonstrate the following claim.

\begin{claim}
   Let $\vct{b}^{(j)}$ be the vector we recover at iteration $j$.
\begin{itemize}
\item The vector $\vct{b}^{(j)}$ contains all but at most $\frac14
  \frac{k}{2^j}$ residual elements of $\signal_k^{(j)}$, with ``good''
  estimates.
\item The vector $\vct{b}^{(j)}$ contains at most $\frac14
  \frac{k}{2^j}$ residual elements of $\signal_k$ with ``bad''
  estimates.
\item The total sum square error over all ``good'' estimates is at
  most
\[
  \left[2 - \left( \frac34 \right)^{j+1} \right]
 -\left[2 - \left( \frac34 \right)^j \right]
 = \frac14 \left( \frac34 \right)^j.
\]
\end{itemize}
\end{claim}

\begin{proof}
To simplify notation, let $T$ be the set of un-recovered elements of
$\signal_k$ at iteration $j$; i.e., the support of
$\signal^{(j)}$.  We know that $|T| \leq k/2^j$.  The proof proceeds in three steps.
\vspace{0.25cm}

\noindent\textbf{Step 1. Isolate heavy hitters with little noise.}
Consider the action of a Bernoulli sign-flip matrix $\mtx{B}\hp \mtx{S}$
with $O(k/2^j)$ rows. From previous work~\cite{CCF,AMS99:Space-Frequency}, it follows that, if constant factors parametrizing the matrices are chosen properly, 
\begin{lemma}
	\label{lemma:isolation}
	For each row $\rho$ of $\mtx{B}$, the following holds with probability 
	$\Omega(1)$:
\begin{itemize}
\item There is exactly one element $t$ of $T$ ``hashed'' by
  $\mtx{B}$; i.e., there is exactly one $t\in T$ with
  $\rho_t=1$.
\item There are $O(N\cdot 2^j/k)$ total positions (out of $N$) hashed
  by $\mtx{B}$.
\item The dot product $(\rho \hp \mtx{S})\vct{r}^{(j)}$ is
  $\mtx{S}_t\vct{r}^{(j)}_t\pm
  O\left(\frac{2^j}{k}\pnorm{2}{\nu^{(j)}_1}\right)$.
\end{itemize}
\end{lemma}
\begin{proof}
	(Sketch.)
For intuition, note that the estimator $\mtx{S}_t(\rho \hp \mtx{S})
\vct{r}^{(j)}$ is a random variable with mean $\vct{r}^{(j)}_t$ and
variance $\pnorm{2}{\nu_1^{(j)}}^2$.  Then the third claim and the first two
claims assert that the expected behavior happens with probability
$\Omega(1)$.
\end{proof}

In our matrix $\mtx{B}^{(j)}$, the number of rows is not $k/2^j$ but
$kc^j$ for some $c$, $1/2 < c < 1$.  Take $c=2/3$.  We obtain a stronger conclusion to the lemma.  The dot product $(\rho \hp \mtx{S})\vct{r}^{(j)}$ is
\[  
       \mtx{S}_t\vct{r}^{(j)}_t\pm
             O\left(\frac1{k(2/3)^j}\pnorm{2}{\nu^{(j)}_1}\right)
       =
      \mtx{S}_t\vct{r}^{(j)}_t\pm
                \frac18\left((3/4)^j\frac{2^j}{k}\pnorm{2}{\nu^{(j)}_1}\right),
\]
provided constants are chosen properly.
Our lone hashed heavy hitter $t$ will dominate the dot product provided
\[  \left|\vct{r}^{(j)}_t\right| \geq  
           \frac18\left((3/4)^j\frac{2^j}{k}\pnorm{2}{\nu^{(j)}_1}\right).
\]
We show in the remaining steps that we can likely recover such heavy
hitters; i.e., {\sc Identify} identifies them and {\sc Estimate}
returns a good estimate of their values.  There are at most $(k/2^j)$
heavy hitters of magnitude less than
$\frac18\left((3/4)^j\frac{2^j}{k}\pnorm{2}{\nu^{(j)}_1}\right)$ which
we will not be able to identify nor to estimate but they contribute a
total of
$\frac18\left((3/4)^j\pnorm{2}{\nu^{(j)}_1}\right)$ 
noise energy to the residual for the next round (which still meets our invariant).
\vspace{0.25cm}

\noindent\textbf{Step 2. Identify heavy hitters with little noise.}
Next, we show how to identify $t$.  Since there are $N/k^{\Theta(1)}$
positions hashed by $\mtx{B}^{(j)}$, we need to learn the
$O(\log(N/k))$ bits describing $t$ in this context.  Previous sublinear algorithms~\cite{CM06:Combinatorial-Algorithms,GSTV07:HHS} used a trivial error correcting code, in which the $t$'th
column was simply the binary expansion of $t$ in direct sum with a
single 1.  Thus, if the signal consists of $\signal_t$ in the $t$'th
position and zeros elsewhere, we would learn $\signal_t$ and
$\signal_t$ times the binary expansion of $t$ (the latter interpreted
as a string of 0's and 1's as real numbers).  These algorithms require strict control on the failure probability of each measurement in order to use such a trivial encoding.  In our case, each measurement succeeds only with
probability $\Omega(1)$ and, generally, fails with probability
$\Omega(1)$.  So we need to use a more powerful error correcting
code and a more reliable estimate of $|x_t|$.

To get a reliable estimate of $|\signal_t|$, we use the
$b=\Theta(\log\log N)$-parallel repetition code of all 1s.  That is, we get $b$
independent measurements of $|\signal_t|$ and we decode by taking the
median.  Let $p$ denote the success
probability of each individual measurement.  Then we expect the
fraction $p$ to be
approximately correct estimates of $|\signal_t|$, we achieve close to
the expectation, and we can arrange that
$p>1/2$.  It follows that the median is approximately correct.  We use
this value to threshold the subsequent measurements (i.e., the bits in
the encoding) to $0/1$ values.

Now, let us consider these bit estimates.  In a single error-correcting code
block of $b=\Theta(\log\log N)$ measurements, we will get close to the
expected number, $bp$, of successful measurements, except with
probability $1/\log(N)$, using the Chernoff bound.  In the favorable
case, we get a number of failures less than the (properly chosen)
distance of the error-correcting code and we can recover the block
using standard nearest-neighbor decoding.  The number of
error-correcting code blocks associated with $t$ is
$O(\log(N/k)/\log\log N)\le O(\log N)$, so we can take a union bound
over all blocks and conclude that we recover $t$ with probability
$\Omega(1)$.  The invariant requires that the failure probability decrease with
$j$.  Because the algorithm takes $O(j)$ parallel independent repetitions, we guarantee that the failure probability decreases with $j$ by taking the union over the repetitions.

We summarize these discussions in the following lemma. We refer to these heavy hitters in the list $\Lambda$ as the $j$-large heavy hitters.
\begin{lemma}
	{\sc Identify} returns a set $\Lambda$ of signal positions that contains at least $3/4$ of the heavy hitters in $T$, $|T| \leq k/2^j$, that have magnitude
at least $\frac18\left((3/4)^j\frac{2^j}{k}\pnorm{2}{\nu^{(j)}_1}\right)$.
\end{lemma}

We also observe that our analysis is consistent with the bounds we give on the additional measurement noise $\nu_2$.  The
permutation matrix $\mtx{P}$ in $\mtx{\Phi}$ is applied before $\nu_2$
is added and then $\mtx{P}^{-1}$ is applied after $\nu_2$ by the
decoding algorithm.  It follows that we can assume $\nu_2$ is permuted
at random and, therefore, by Markov's inequality, each measurement gets at most an
amount of noise energy proportional to its fair share of
$\pnorm{2}{\nu_2}^2$.  Thus, If there are $m=\Theta(k\log N/k)$
measurements, each measurement gets $\frac{\pnorm{2}{\nu_2}^2}m$ noise
energy and identification succeeds anyway provided the lone heavy
hitter $t$ in that bucket has square magnitude at least
$\frac{\pnorm{2}{\nu_2}^2}m$, so the at most $k$ smaller heavy
hitters, that we may miss, together contribute energy
$\frac{k\pnorm{2}{\nu_2}^2}m = 
      O\left(\frac{\pnorm{2}{\nu_2}^2}{\log(N/k)}\right)$.  
If we recall the definition and value of $\pnorm{2\leadsto 2}{\mtx{\Phi}}$, we see that this error meets our bound.
\vskip 2pc

\noindent\textbf{Step 2. Estimate heavy hitters.}
Many of the details in this step are similar to those in Lemma~\ref{lemma:isolation} (as well as to previous work as the function {\sc Estimate} is essentially the same as {\sc Count Sketch}), so we give only a brief summary.

First, we discuss the failure probability of the {\sc Estimate}
procedure.  Each estimate is a complete failure with probability
$1-\Omega(1)$ and the total number of identified positions is
$O\left(jk\left(2/3\right)^j\right)$.  Because we perform $j$
parallel repetitions in estimation, we can easily arrange to
lower that failure probability, so we
assume that the failure probability is at most
$\Theta\left(\left(3/4\right)^j\right)$, and that we get
approximately the expected number of (nearly) correct estimates.
There are $k(2/3)^j$ heavy hitters in $\Lambda$, so
the expected number of failures is $(1/4)(k/2^j)$.
These, along with the at most $1/4(k/2^j)$ missed $j$-large
heavy hitters, will form  $\signal^{(j+1)}$, the at-most-$k/2^{j+1}$
residual heavy hitters at the next iteration.

In iteration $j$, {\sc Identity} returns a list $\Lambda$ with
$k(2/3)^j$ heavy hitter position identified.  A group of $k(2/3)^j$
measurements in $\mtx{E}^{(j)}$ yields estimates for the positions in
$\Lambda$ with aggregate $\ell_2$ error $\pm O(1)$, additively.  An
additional $O\left(\left(4/3\right)^j\right)$ times more measurements,
$O(k(8/9)^j)$ in all,
improves the estimation error to $(1/8)\left(3/4\right)^j$,
additively.
These errors, together with the omitted heavy hitters that are not $j$-large and
$\nu^{(j)}$ form the new noise vector at the next iteration, $\nu^{(j+1)}$.

Finally, consider the effect of $\nu_2$.  We would like to argue that, as in the identification step, the noise vector $\nu_2$ is permuted at random and each measurement is
corrupted by $\frac{\pnorm{2}{\nu_2}^2}m$, where $m=\Theta(k\log(N/k))$ is the
number of measurements, approximately its fair share of $\pnorm{2}{\nu_2}^2$.
Unfortunately, the contributions of $\nu_2$ to the various measurements are not
independent as $\nu_2$ is permuted, so we cannot use such a simple analysis.  Nevertheless, they are negatively correlated and we can
achieve the result we want using~\cite{Dubhashi96negativedependence}.
The total $\ell_2$ squared error of the corruption over
all $O(k)$ estimates is $\pnorm{2}{\nu_2}^2/\log(N/k)$, which
will meet our bound.  That is, since  $\pnorm{2\leadsto 2}{\mtx{\Phi}}^2
= O(\log^2(k)\log(N/k))$, the $\nu_2$ contribution to the error
is
\[  
   O\left(\frac{\pnorm{2}{\nu_2}}{\sqrt{\log N/k}}\right) 
=  O\left( \frac{ \log(k)\pnorm{2}{\nu_2}}{\pnorm{2\leadsto 2}{\mtx{\Phi}}}\right),
\] 
as claimed, whence we read off the factor, $\log(k)$ (improvable to
$\log^{1/2+o(1)} k$), which is
directly comparable to other results that scale $\mtx{\Phi}$
properly.

\end{proof}

\subsection{Efficiency}

\subsubsection{Number of Measurements}

The analysis of isolation and estimation matrices are similar; the
number of measurements in isolation dominates.

The number of measurements in iteration $j$ is computed as follows.
There are $O(j)$ parallel repetitions in iteration $j$.  They each
consist of $k(2/3)^j$ measurements arising out of $\mtx{B}^{(j)}$ for
identification times $O(\log(N/k))$ measurements for the error
correcting code, plus $k(2/3)^j$ times $O((4/3)^j)$ for estimation.
This gives
\[\Theta\left(jk\left(\frac23\right)^j\log(N/k)+jk\left(\frac89\right)^j\right)
=k\log(N/k)\left(\frac89+o(1)\right)^j.\]
Thus we have a sequence bounded by a geometric sequence with ratio
less than 1.  The sum, over all $j$, is $O(k\log(N/k))$.

\subsubsection{Encoding and Update Time}

The encoding time is bounded by $N$ times the number of non-zeros in
each column of the measurement matrix.  This was analyzed above in
Section~\ref{sec:proof}; there are $\log^2(k)\log(N/k)$ non-zeros per
column, which is suboptimal by the factor $\log^2(k)$.  By comparison,
some proposed methods use dense matrices, which are suboptimal by the
exponentially-larger factor $k$.  
This can be improved slightly, as follows.
Recall that we used $j$ parallel
repetitions in iteration $j,$ $j<\log(k)$, to make the failure
probability at iteration be; e.g., $2^{-j}$, so the sum over $j$
is bounded.  We could instead use failure probability $1/j^2$, so that
the sum is still bounded, but the number of parallel repetitions will
be $\log(j)$, for $j\le\log(k)$.  This results in
$\log(k)\log\log(k)\log(N/k)$ non-zeros per column and $\nu_2$
contribution to the noise equal to
$\sqrt{\log(k)\log\log(k)}
\frac{\pnorm{2}{\nu_2}}{\pnorm{2\leadsto 2}{\mtx{\Phi}}}$.

We can use a pseudorandom number generator such as
$i\mapsto \left\lfloor(ai+b \bmod d)/B\right\rfloor$ for random $a$
and $b$, where $B$ is the number of buckets.  Then we can, in time
$O(1)$, determine into which bucket any $i$ is mapped and determined
the $i$'th element in any bucket.

Another issue is the time to find and to encode (and to decode) the
error-correcting code.  Observe that the length of the code is
$O(\log\log N)$.  We can afford time exponential in the length,
{\em i.e.}, time $\log^{O(1)}N$, for finding and decoding the code.
These tasks are straightforward in that much time.

\subsubsection{Decoding Time}

As noted above, we can quickly map positions to buckets and find the
$i$'th element in any bucket, and we can quickly decode the
error-correcting code.  The rest of the claimed runtime is
straightforward.

\section{Conclusion}
In this paper, we construct an approximate sparse recovery system that is essentially optimal: the recovery algorithm is a sublinear algorithm (with near optimal running time), the number of measurements meets a lower bound, and the update time, encode time, and column sparsity are each within $\log$ factors of the lower bounds.  We conjecture that with a few modifications to the distribution on measurement matrices, we can extend this result to the $\ell_1 \leq C \ell_1$ error metric guarantee.  We do not, however, think that this approach can be extended to the ``for all'' signal model (all current sublinear algorithms use at least one factor $O(\log N)$ additional measurements) and leave open the problem of designing a sublinear time recovery algorithm and a measurement matrix with an optimal number of rows for this setting.
\bibliography{foreachbib}
\bibliographystyle{plain}   

\section{Appendix}

We have a full description of the matrix algebra defined in Table~\ref{table:matrixalgebra}.
\begin{itemize}
	\item \textbf{Row direct sum.}  The row direct sum $\mtx{A} \rds \mtx{B}$ is a matrix with $N$ columns that is the vertical concatenation of $\mtx{A}$ and $\mtx{B}$. 
	\item \textbf{Element-wise product.} If $\mtx{A}$ and $\mtx{B}$ are both $r \times N$ matrices, then $\mtx{A} \hp \mtx{B}$ is also an $r \times N$ matrix whose $(i,j)$ entry is given by the product of the $(i,j)$ entries in $\mtx{A}$ and $\mtx{B}$.
	\item \textbf{Semi-direct product.} Suppose $\mtx{A}$ is a matrix of $r_1$ rows (and $N$ columns) in which each row has exactly $h$ non-zeros and $\mtx{B}$ is a matrix of $r_2$ rows and $h$ columns. Then $\mtx{B}\sdp \mtx{A}$ is the matrix with $r_1r_2$ rows, in which each non-zero entry $a$ of $\mtx{A}$ is replaced by $a$ times the $j$'th column of $\mtx{B}$, where $a$ is the $j$'th non-zero in its row. This matrix construction has the following interpretation. Consider $(\mtx{B} \sdp \mtx{A})\signal$ where $\mtx{A}$ consists of a single row, $\mtx{\rho}$, with $h$ non-zeros and $\signal$ is a vector of length $N$. Let $\mtx{y} = \mtx{\rho} \hp \signal$ be the element-wise product of $\rho$ and $\signal$. If $\mtx{\rho}$ is 0/1-valued, $\mtx{y}$ picks out a subset of $\signal$. We then remove all the positions in $\mtx{y}$ corresponding to zeros in $\mtx{\rho}$, leaving a vector $\mtx{y'}$ of length $h$. Finally, $(\mtx{B}\sdp\mtx{A})\signal$ is simply the matrix-vector product $\mtx{B} \mtx{y'}$, which, in turn, can be interpreted as selecting subsets of $\mtx{y}$, and summing them up. Note that we can modify this definition when $\mtx{A}$ has fewer than $h$ non-zeros per row in a straightforward fashion. 
\end{itemize}

\begin{center}
\begin{figure}[h!]
\mbox{
  \begin{minipage}[htbf]{6.8in}
\tt 
\begin{tabular}{|p{0.9\textwidth}|}
\hline
\centerline{{\sc Recover}$(\mtx{\Phi},\vct{y})$} \\
Output: $\widehat x =$ approximate representation of $x$ \\ \\
$\vct{y}=\mtx{P}^{-1}\vct{y}$\\
$\vct{a}^{(0)} = 0$ \\
For $j = 0$ to $O(\log k)$ $\{$ \\
\hspace{2pc} $\vct{y} = \vct{y} - \mtx{P}^{-1}\mtx{\Phi}\vct{a}^{(j)}$ \\
\hspace{2pc} split $\vct{y}^{(j)} = \vct{w}^{(j)} \rds \vct{z}^{(j)}$ \\
\hspace{2pc} $\Lambda =$ \sc{Identify}$(\mtx{D}^{(j)}, \vct{w}^{(j)})$ \\
\hspace{2pc} $\vct{b}^{(j)} =$ \sc{Estimate}$(\mtx{E}^{(j)}, \vct{z}^{(j)}, \Lambda)$ \\
\hspace{2pc} $\vct{a}^{(j+1)} = \vct{a}^{(j)} + \vct{b}^{(j)}$ \\
$\}$ \\
$\widehat x = \vct{a}^{(j)}$ \\ \\
\hline \\
\centerline{{\sc Identify}$(\mtx{D}^{(j)}, \vct{w}^{(j)})$} \\
Output: $\Lambda =$ list of positions \\ \\
$\Lambda = \emptyset$ \\
Divide $\vct{w}^{(j)}$ into sections $[\vct{v},\vct{u}]$ of size $O(\log(c^j (N/k)))$ \\
For each section  $\{$ \\
\hspace{2pc} $u = {\rm median} (|\vct{v}_\ell|)$ \\
\hspace{2pc} For each $\ell$  \hfill // threshold measurements \\
\hspace{4pc} $\vct{u}_\ell = \Theta(\vct{u}_\ell - u/2)$ \hfill // $\Theta(u) = 1$ if $u > 0$, $\Theta(u) = 0$ otherwise \\
\hspace{2pc} Divide $\vct{u}$ into blocks $b_i$ of size $O(\log\log N)$\\
\hspace{2pc} For each $b_i$ \\
\hspace{4pc} $\beta_i$ = {\sc Decode}$(b_i)$ \hfill // using error-correcting code \\
\hspace{2pc} $\lambda = ${\sc Integer}$(\beta_1,\beta_2,\ldots)$
                 \hfill // integer rep'ed by bits $\beta_1,\beta_2,\ldots$ \\
\hspace{2pc} $\Lambda = \Lambda \cup \{\lambda\}$ \\
$\}$ \\ \\
\hline \\
\centerline{{\sc Estimate}$(\mtx{E}^{(j)}, \vct{z}^{(j)}, \Lambda)$} \\
Output: $\vct{b} =$ vector of positions and values \\ \\
$\vct{b} = \emptyset$ \\
For each $\lambda \in \Lambda$  \\
\hspace{2pc} $\vct{b}_\lambda = {\rm median}_{\ell {\rm ~s.t. } \mtx{B}_{\ell,\lambda}^{(j)}=1} (\vct{z}_\ell^{(j)} \mtx{S}_{\ell,\lambda}^{(j)})$ \\
\hline
\end{tabular}
\end{minipage} }
\caption{Pseudocode for the overall decoding algorithm.}
\label{fig:decodingalg}
\end{figure}
\end{center}


\end{document}